\pdfoutput=1
\documentclass[reprint,amsmath,amssymb,aps, physrev]{revtex4-2}

\usepackage[dvipsnames]{xcolor}

\usepackage[english]{babel}
\usepackage{mathtools}

\usepackage{amsthm}

\usepackage[colorlinks,linkcolor=darkblue,citecolor=darkblue,linktocpage,hypertexnames=true,pdfpagelabels]{hyperref}
\usepackage[capitalise]{cleveref}

\newtheorem{theorem}{Theorem}
\newtheorem{corollary}[theorem]{Corollary}
\newtheorem{lemma}[theorem]{Lemma}

\newcommand{\mc}{\mathcal}
\newcommand{\mr}{\mathrm}

\definecolor{darkblue}{rgb}{0.0, 0.0, 0.55}

\begin{document}

\title{Universal spin models are universal approximators in machine learning}

\author{Tobias Reinhart}
\email{tobias.reinhart@uibk.ac.at}
\affiliation{Institute for Theoretical Physics, University of Innsbruck,\\ Technikerstr.\ 21a,\ A-6020 Innsbruck, Austria}

\author{Gemma De les Coves}
\affiliation{Departament d'Enginyeria, Universitat Pompeu Fabra, Carrer T\`anger 122, 08018 Barcelona, Catalonia}
\affiliation{ICREA, Instituci\'o Catalana d'Estudis i Recerca Avan\c{c}ats, Passeig Llu\'is Companys 23, 08010 Barcelona, Catalonia }

\begin{abstract}
One of the theoretical pillars that sustain certain machine learning models are universal approximation theorems, which prove that they can approximate all functions from a function class to arbitrary precision. Independently, classical spin models are termed universal if they can reproduce the behavior of any other spin model in their low energy sector. Universal spin models have been characterized via sufficient and necessary conditions, showing that simple models such as the 2d Ising with fields are universal. In this work, we prove that universal spin models are universal approximators of probability distributions. This enables us to leverage the characterization of the former to reveal conditions which are sufficient for universal approximation. Deriving universal approximation theorems thus amounts to verifying these conditions, yielding a unified recipe for universal approximation theorems applicable to a wide range of models. We explicitly test this recipe for restricted and deep Boltzmann machines, as well as for deep belief networks. This work illustrates that independently discovered universality statements may be intimately related, enabling the transfer of results. 
\end{abstract}

\maketitle

Spin models are essentially generalizations of the Ising model \cite{Is25}. 
They are toy models of complex systems that consist of many interacting, classical degrees of freedom, called spins. 
Their versatile nature has made them amenable for problems not only in physics, but also in computational complexity, biology or the social sciences (see e.g.\ \cite{Re21c} for a list of applications), as well as machine learning. 
In the latter, they give rise to energy based models, such as restricted Boltzmann machines (RBMs) or deep Boltzmann machines (DBMs) \cite{Os19}, where one is interested in  the Boltzmann distribution, which describes the behavior of the spin model at thermal equilibrium.
This probability distribution is parameterized by couplings that represent  interactions between spins. The learning task consists of finding parameters such that the output  distribution in the so-called visible spins approximates a given target distribution. 

Several machine learning models, including RBMs and DBMs, are theoretically supported by universal approximation theorems (UATs), which prove that they can approximate any function from a given class to arbitrary precision \cite{LR08, Ro10, Mo15, Mo17}.
In particular, the UATs for RBMs and DBMs guarantee that both can approximate arbitrary probability distributions. In this sense, they are universal approximators of probability distributions.

Independently, certain spin models
(such as the 2d Ising model with fields or the 3d Ising model) have been proven to be universal, meaning that their low energy sector can reproduce the behavior of any other spin model \cite{De16b, Re24}. 
While universal spin models have been fully characterized in terms of sufficient and necessary conditions, to the best of our knowledge no such characterization exists for universal approximators in machine learning (cf.\ \cite{Mo17}). 

\begin{figure}[t]
    \centering
    \includegraphics{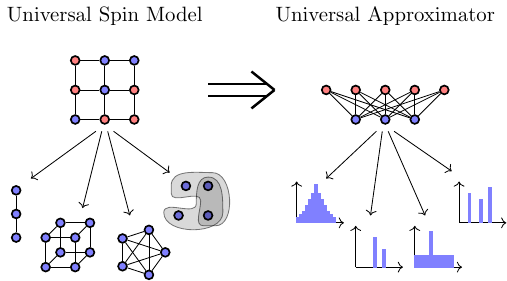}
    \caption{ 
    We prove that universal spin models are universal approximators of probability distributions in machine learning. 
    We leverage the characterization of the former \cite{Re24} to derive UATs for the latter. We explicitly do so for RBMs, DBMs and DBNs. 
    }
    \label{fig:main}
\end{figure}

In this paper, we prove that universal spin models are universal approximators of probability distributions by leveraging  \cite{Re24} (see \cref{fig:main}). It follows that the characterization of the former extends to the latter, yielding a characterization of universal approximators for energy based models \footnote{That is machine learning models, defined by a Hamiltonian the corresponding Boltzmann distribution, examples are DBMs and RBMs.} 
and in turn a unified recipe for the derivation of UATs. 
We apply this recipe to RBMs and DBMs, resulting in proofs of the respective UATs. As a first step beyond energy based models, we apply the procedure to derive a UAT for deep belief networks (DBNs). Overall, we find that universality statements of spin models and some machine learning models are intimately connected.

\emph{Spin systems, spin models and simulations}.---
We shall consider the following simplified versions of definitions of \cite{Re24}. 
\emph{Spin models} are potentially infinite sets of spin systems. 
\emph{Spin systems} consist of finitely many classical, binary degrees of freedom, or `spins' interacting according to a hypergraph $(V,E)$. Each spin is associated to a vertex $v\in V$, and each interaction between spins is associated to a hyperedge  $e\in E$ and given by a local energy function. 
A configuration $\mathbf{s}$ assigns a number from $\{0,1\}$ to each spin. 
Mapping configurations $\mathbf{s}$ to the sum of their local energy contributions $H(\mathbf{s})$ defines the Hamiltonian $H$ of the spin system. 
The Boltzmann distribution, 
\begin{equation}
    p(\mathbf{s}) = \frac{1}{Z}e^{-H(\mathbf{s})},
\end{equation}
is the probability distribution of configurations at thermal equilibrium (at inverse temperature $\beta=1$), where the normalization $Z$ is the partition function.

\emph{Simulations} are transformations of spin systems that preserve the low energy sector. 
Given spin systems $S$ and $T$ (for source and target) with Hamiltonians $H_S$ and $H_T$ and with $V_T \subseteq V_S$ and a cut-off energy $\Delta$, we say that $S$ simulates $T$ with cut-off $\Delta$ (and write $S \to T$) if up to a constant energy shift $\Gamma$ (\cref{fig:sim})
\begin{enumerate}
    \item for each target configuration $\mathbf{t}$ with $H_T(\mathbf{t})\leq \Delta$ there exists a unique configuration $\mathbf{h}_{\mathbf{t}}$ of auxiliary (or hidden) spins $V_S\setminus V_T$ such that the combined source configuration 
    $(\mathbf{t},\mathbf{h}_{\mathbf{t}})$ satisfies $H_S(\mathbf{t},\mathbf{h}_{\mathbf{t}}) = H_T(\mathbf{t})$, and 
    \item all other source configurations have energy at least $\Delta$.
\end{enumerate}
\begin{figure}[t]
    \centering
     \includegraphics{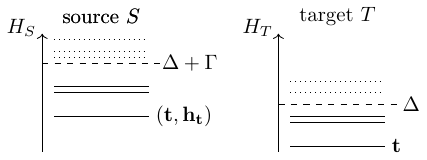}
    \caption{If $S$ simulates $T$,
    $H_S(\mathbf{t},\mathbf{h}_{\mathbf{t}}) = H_T(\mathbf{t})$
    below the cut-off $\Delta$ up to a shift $\Gamma$.}\label{fig:sim}
\end{figure}  

For a simple example consider any spin system $T$ and any spin $i$ from $T$, and construct $S$ by adding a single auxiliary spin $i'$ that interacts with $i$ with a large ferromagnetic coupling $\Delta$. 
Then, for any target configuration $\mathbf{t}$, taking $i'$ to be in the same state as $i$ yields a source configuration with equal energy (up to a shift). 
Conversely, all other source configurations, i.e.\ where $i$ and $i'$ are in different states, yield energies above the cut-off.
This simulation allows one to reduce the number of adjacent spins of $i$ in $T$. 
\begin{equation}\label{eq:example sim}
    \centering
    \includegraphics{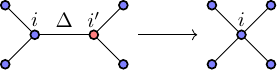}
    \end{equation}
Applying it repeatedly, one can prove that spin systems on 2d square lattices can simulate those on arbitrary planar graphs \cite{Re24}.

We say that $S$ simulates a real-valued function $f$ of spin configurations if it simulates any spin system with Hamiltonian $f$. 
We say that $S$ simulates a probability distribution $p$ over spin configurations on $V_T$ if for all such configurations $\mathbf{t}$, 
\begin{equation}\label{eq:sim preserve Boltzmann}
     p_{V_T}[S](\mathbf{t}) = p(\mathbf{t})+\mathcal{O}(e^{-\Delta}),
\end{equation} 
that is, $p_{V_T}[S]$, the marginal of the source Boltzmann distribution $p[S]$ over physical spins $V_T$, approximates the target distribution $p$ with an error that decreases exponentially with the cut-off. 

Simulations also approximately preserve several properties of spin systems \cite{Re24}, 
illustrating that they can be interpreted as approximations whose error is controlled by the cut-off.
In particular, they preserve Boltzmann distributions (Lemma 2 of \hyperref[sec:supplemental]{the supplemental material}), that is, whenever $S \to T$ then $S$ also simulates the Boltzmann distribution of $T$ (according to \cref{eq:sim preserve Boltzmann}).
Note that the pointwise approximation of \cref{eq:sim preserve Boltzmann} implies approximations in terms of the Kullbach--Leibler divergence and the total variation distance. 
Further, simulations have a modular structure, meaning they can  be combined in several ways, allowing to construct complicated simulations from simpler ones
\footnote{
Simulations can be composed, scaled and added \cite{Re24}. This implies, e.g.\ that they can be constructed from \emph{gadgets}, spin systems that locally replace a  part of the target, similar to \cref{eq:example sim}.}

Simulations are easily promoted to spin models. We say that a spin model $\mc{M}$ simulates a spin system $T$ (or a function $f$ or a probability distribution $p$) 
with cut-off $\Delta$ if there exists a spin system $S \in \mc{M}$ that simulates $T$ (or $f$, or $p$) with cut-off $\Delta$. We say that $\mc{M}$ simulates $T$ (or $f$, or $p$) if this holds for all cut-offs. 
Finally, a spin model is called \emph{universal} if it simulates all spin systems. 
It is called a \emph{universal approximator} (of probability distributions) if it simulates all probability distributions over spin configurations.

Universal spin models can be characterized in terms of seemingly weak properties, implying that fairly simple models like the 2d Ising model with fields or the 3d Ising model are universal \cite{Re24}.
Here we characterize them by flag completeness and closure. $S$ is a \emph{flag system}  for configuration $\mathbf{x}$ if up to an energy shift,
\begin{equation}\label{eq:flag system}
    H_S(\mathbf{s}, \mathbf{h}) = \begin{cases}
        0 \ &\text{if}\  \mathbf{s}= \mathbf{x} \ \text{and} \ \mathbf{h}= \mathbf{h}_{\mathbf{x}} \\
        1 \ &\text{if} \ \mathbf{s}\neq \mathbf{x} \ \text{and} \ \mathbf{h}= \mathbf{h}_{\neq \mathbf{x}} \\
        \geq 1 \ &\text{if} \ \mathbf{s}= \mathbf{x} \ \text{and} \ \mathbf{h}\neq \mathbf{h}_{\mathbf{x}} \\
         \geq 2 \ &\text{if} \ \mathbf{s}\neq \mathbf{x} \ \text{and} \ \mathbf{h}\neq \mathbf{h}_{\neq \mathbf{x}}.
    \end{cases}
\end{equation}
In words, given a fixed configuration $\mathbf{s}$ of physical spins, the energy of $S$ is minimal whenever 
the configuration $\mathbf{h}$ of flag spins
correctly signals whether $\mathbf{s} = \mathbf{x}$ or not, i.e.\ 
if $\mathbf{h}$ equals the truth value of ``$\mathbf{s} = \mathbf{x}$" with configurations $\mathbf{h}_{\mathbf{x}}$ corresponding to ``TRUE", and $\mathbf{h}_{\neq \mathbf{x}}$ corresponding to ``FALSE".
Thus, in the low energy sector, information over the physical spins can be retrieved from the flag-spins, only.  
See \cref{eq:flag RBM} and \cref{eq: flag sys 2} for examples.

A spin model is \emph{flag complete} if it contains flag systems for arbitrary configurations. Additionally, a spin model $\mc{M}$ is \emph{closed} for a subset $\mc{N}\subset \mc{M}$ if it simulates arbitrary non-negative linear combinations from $\mc{N}$. 
\begin{theorem}\label{thm:univ}
    A spin model $\mc{M}$ is universal if and only if  
    \begin{enumerate}
    \item it is flag complete, and 
    \item it is closed for any subset of flag systems with disjoint flag spins.
    \end{enumerate}
\end{theorem}
\begin{proof}
    The ``only if" direction is immediate, since then $\mc{M}$ simulates all spin systems.
    
    We prove the ``if" direction constructively. 
    Without loss of generality we assume that all flag systems have zero energy shift, i.e.\ satisfy \cref{eq:flag system} exactly.
    
    First, we prove that linear combinations of flag systems can simulate an arbitrary function $f$. The idea is that $\lambda S_{\mathbf{x}}$ for large $\lambda$ moves $\mathbf{x}$ to the low energy sector. 
    Thus, the linear combination $S=\sum_i \lambda_i S_{\mathbf{x}_i}$ over all low energy target configurations $f(\mathbf{x}_i)\leq \Delta$, for appropriate $\lambda_i$,
    reproduces the low energy sector of $f$.
    Specifically, we define 
    \begin{equation}
    \begin{split}
        f_{\mr{diff}} &= \mr{max}\{ f(\mathbf{x}_i)-f(\mathbf{x}_j) \mid f(\mathbf{x}_i), f(\mathbf{x}_j)\leq \Delta\}\\
        \lambda_i &= \Delta + f_{\mr{diff}} - f(\mathbf{x}_i) \\
        \Gamma &= \sum_i \lambda_i - (\Delta+f_{\mr{diff}}).
    \end{split}
    \end{equation}
    Then,
    high energy target configurations $\mathbf{y}$ yield $ H_S(\mathbf{y},\mathbf{h}) \geq \Delta + \lambda$, regardless of the configuration of flag spins, $\mathbf{h}$.
    Low energy target configurations $\mathbf{x}_i$, yield 
    $H_S(\mathbf{x}_i,\mathbf{h})=  \Gamma + f(\mathbf{x}_i)$
    if the flag spins flag correctly, and 
    $H_S(\mathbf{x}_i,\mathbf{h}) \geq \Gamma + \Delta$ otherwise. 
    Thus, $S$ simulates $f$ with cut-off $\Delta$ and energy shift $-\Gamma$.

    Now any spin system $T$ can be decomposed into the sum of its local energy functions $T= \sum_e T_e$, 
    each of which can be simulated by a linear combination of flag systems $\sum_{i_e}\lambda_{i_e} S_{\mathbf{x}_{i_e}}$. 
    Additivity of simulations \cite{Re24} then implies that 
    \begin{equation}\label{eq:sim proof thm1}
        \sum_e \sum_{i_e} \lambda_{i_e}  S_{\mathbf{x}_{i_e}} \to \sum_e T_e = T.
    \end{equation}
    Finally, closure and flag completeness imply that $\mc{M}$ simulates the left hand side of \eqref{eq:sim proof thm1}. Since simulations can be composed, this  proves the claim (see Fig. 2 of \cite{Re24} for an illustration). 
\end{proof}
Constructing the simulation of $T$ as a sum of simulations of its local energy functions $T_e$  is more efficient than directly simulating $T$, as the former requires $\mathcal{O}(\sum_e 2^{\vert e \vert})$ instead of $\mathcal{O}(2^{\vert V_T \vert})$ auxiliary spins, which, for typical targets $T$, is polynomially instead of exponentially many.

\emph{Restricted and deep Boltzmann machines}.---
We now apply \cref{thm:univ} to certain machine learning models, more precisely to their underlying spin model. We will prove that for both RBMs and DBMs the underlying spin model is universal. 

For our purposes, we identify RBMs with their underlying spin model. In our eyes, an RBM is a spin system with complete bipartite interaction graph between visible spins $V_v$ and hidden spins $V_h$, with Hamiltonian
\begin{equation}\label{eq:RBM}
    H(\mathbf{v}, \mathbf{h}) = \mathbf{v}^t \mathbf{b}_v + \mathbf{h}^t \mathbf{b}_h + \mathbf{v}^t \mathbf{W} \mathbf{h},
\end{equation}
parameterized by coupling vectors $\mathbf{b}_v, \mathbf{b}_h$ and coupling matrix $\mathbf{W}$, where ${}^t$ denotes transpose. 

Similarly, a DBM is a spin system whose interaction graph consists of fully connected layers $V=\cup_iV_i$, 
such that each subsystem $R_i$ with spins $V_i\cup V_{i+1}$ defines an RBM, with Hamiltonian $H_i$. 
Denoting the configuration of layer $i$ by $\mathbf{h}_i$, 
the total Hamiltonian is given by
\begin{equation}\label{eq: Ham DBM}
    H(\mathbf{h}_0, \ldots, \mathbf{h}_k) 
     = \left( \sum_{i=0}^{k-2} H_i(\mathbf{h}_i,\mathbf{h}_{i+1}) \right)
     + H_{k-1}(\mathbf{h}_{k-1},\mathbf{h}_k).
\end{equation}
We identify $\mathbf{h}_0 = \mathbf{v}$ and call spins from $V_0$ visible spins, 
and those from $V\setminus V_0$ hidden spins.

\begin{theorem}\label{thm:RBM univ}
    RBMs are universal. 
\end{theorem}

That is, the spin model given by all RBMs can simulate all spin systems.
Moreover, the physical spins of each simulation are the visible spins of the RBM, and the auxiliary spins are the hidden spins. 
In addition, simulating a function with $k$ low energy configurations requires an RBM with $k$ hidden spins. 

\begin{proof}
    We exploit \cref{thm:univ}. 
    Flagging configuration $\mathbf{x}$ can be achieved by an RBM
    \begin{equation}\label{eq:flag RBM}
        \centering
    \includegraphics{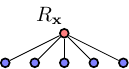}
    \end{equation}
    with a single hidden spin serving as flag spin, and with couplings such that components of $\mathbf{W}$ are given by $W_i=-2$ if $x_i=1$ and $W_i=+2$ if $x_i=0$,
    $\mathbf{b}_v = \mathbf{0}$ and  $\mathbf{b}_h=2\mathrm{d}(\mathbf{x},\mathbf{0})-1$, where $\mr{d}(\cdot,\cdot)$ is the Hamming distance \footnote{We denote by $\mathbf{0}$ the all zero and by $\mathbf{1}$ the all one vector.}.
    It is easy to check that $R_{\mathbf{x}}$ satisfies \cref{eq:flag system} up to a shift of $+1$, with $h_{\mathbf{x}}=+1$ and $h_{\neq \mathbf{x}}=-1$.
  
    Closure is trivial since scaling flag RBMs amounts to changing their couplings, and sums of flag RBMs with disjoint flag spins are trivially RBMs. 
    Following \cref{thm:univ}, the universal RBM is a linear combination of flag RBMs $R_{\mathbf{x}_i}$, that is 
    \begin{equation}
         \centering
    \includegraphics{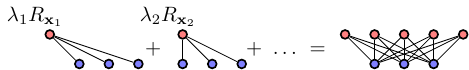}
    \end{equation}
\end{proof}

\begin{theorem}\label{thm:DBM univ}
    DBMs of constant width are universal. 
\end{theorem}

Specifically, the spin model of all DBMs with visible layer of width $n$ and hidden layers of width $n +1$, for any $n$, can simulate all spin systems. 
Again, physical spins of the simulations precisely correspond to visible spins of the DBM.
And simulating a function with $k$ low energy configurations of $n$ spins requires a DBM with $k$ hidden layers, and thus $(n+1)k$ hidden spins.

\begin{proof}
    We exploit \cref{thm:univ}.
    First, define the RBM
    \begin{equation}
        \centering
    \includegraphics{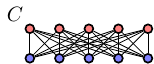}
    \end{equation}
    with $\mathbf{b}_v = \mathbf{b}_h = \mathbf{1}$ and  $\mathbf{W} = -2\: \mathbf{I}$, where $\mathbf{I}$ denotes the identity matrix. Then 
    $H_C(\mathbf{v},\mathbf{h}) = \mathrm{d}(\mathbf{v},\mathbf{h})$.
   In words, in the ground state, $C$  copies configurations from  visible to hidden spins. 
   Defining $F_{\mathbf{x}}$ by
   \begin{equation}\label{eq: flag sys 2}
   \centering
    \includegraphics{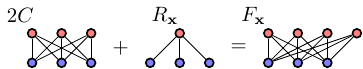}
   \end{equation}
    yields a flag system, i.e.\ satisfies \cref{eq:flag system}.

    Closure follows from simulating 
    the sum of flag systems $F_{\mathbf{x}_1}+F_{\mathbf{x}_2}$ as
    \begin{equation}
        \centering 
    \includegraphics{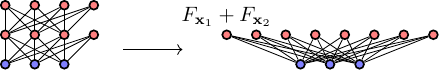}
    \end{equation}
    This simulation works 
    since in the low energy sector, $C$ copies configurations.
    Its cut-off can be controlled by the scaling the copy systems of the left-hand side.
    Note that in contrast to \cref{thm:RBM univ}, closure is non-trivial since sums of flag systems $F_{\mathbf{x}_i}$  do not yield DBMs of constant width.

    Following \cref{thm:univ}, the universal DBM is a simulation of a linear combination of flag systems $F_{\mathbf{x}_i}$, that is, 
    \begin{equation}
         \centering
    \includegraphics{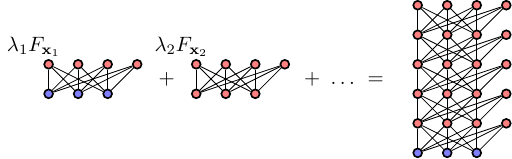}
    \end{equation}

\end{proof}

\emph{Universal spin models are universal approximators}.---
Let us now show that universal spin models are universal approximators of probability distributions. This means that universal spin models, via their Boltzmann distributions, can approximate arbitrary probability distributions over spin configurations to arbitrary precision.

Observe first that any probability distribution $p$ on $\{0,1\}^V$ can be simulated by a spin system $T$.
To this end, construct $T$ such that configurations $\mathbf{t}$ with non-zero probability $p(\mathbf{t})>0$ are assigned low energy, specifically  $-\log(p(\mathbf{t}))$  
and those with zero probability are assigned high energy, specifically $\Delta$. 
Details are provided in Lemma 3 of \hyperref[sec:supplemental]{the supplemental material}.

\begin{theorem}\label{thm:univ approx}
Universal spin models are universal approximators.
\end{theorem}
\begin{proof}
    Consider any probability distribution $p$ and the corresponding spin system $T$ that simulates $p$.
    By assumption a universal spin model can simulate $T$ for arbitrary cut-offs. 
    Simulations preserve Boltzmann distributions (Lemma 2 of \hyperref[sec:supplemental]{the supplemental material}), hence the claim. 
\end{proof}

\begin{corollary}\label{cor:univ approx thm}
    Both RBMs and DBMs 
    are universal approximators. 
\end{corollary}

\begin{proof}
It follows from \cref{thm:RBM univ}, \cref{thm:DBM univ} and  \cref{thm:univ approx}. 
\end{proof}

This is a new proof of the UAT for RBMs and DBMs.
While existing proofs of UATs are derived in a case-by-case fashion
 (cf.\ \cite{LR08,Mo15}), 
ours emanate from an underlying, unified recipe for the derivation for UATs (\cref{thm:univ} and \cref{thm:univ approx}) based on conditions that trigger the onset of universality.   
As such, this work sheds light on the origin and reach of these forms of universality.

\emph{Deep belief networks}.---
Let us extend the characterization of universal approximators to deep Belief Networks (DBNs) and thus beyond energy based models.
Similar to DBMs, DBNs are defined on layerwise fully connected graphs.  However, their output is the marginal of 
\begin{equation}\label{eq:DBN Ham}
    q(\mathbf{h}_0, \ldots, \mathbf{h}_k) = 
    \left( \prod_{i=0}^{k-2} p_i(\mathbf{h}_i \, \vert \,  \mathbf{h}_{i+1})  \right) p_{k-1}(\mathbf{h}_{k-1},\mathbf{h}_k)  
\end{equation}
over visible spins, 
where $p_i(\mathbf{h}_i,\mathbf{h}_{i+1})$ is the Boltzmann distribution of the $i$-th RBM 
 (cf.\ \cref{eq: Ham DBM}).
For $k=2$ ($S=R_0+R_1$) this marginal can be expressed as
\begin{equation}\label{eq:2 DBN}
    q_{V_0}[S](\mathbf{h}_0) = \sum_{\mathbf{h}_1}p[R_0](\mathbf{h}_0 \, \vert \, \mathbf{h}_1) p_{V_1}[R_1](\mathbf{h}_1).
\end{equation}
Following \cite{Mo15}, the corresponding DBM output equals
\begin{equation}\label{eq:2 DBM}
     p_{V_0}[S](\mathbf{h}_0) = \sum_{\mathbf{h}_1}p[R_0](\mathbf{h}_0 \, \vert \, \mathbf{h}_1) (p_{V_1}[R_0]\ast p_{V_1}[R_1])(\mathbf{h}_1),
\end{equation}
where $\ast$ denotes the pointwise, normalized product of probability distributions (see \cite{Mo15}).

Thus, the difference between the DBN and DBM output is the  factor $p_{V_1}[R_0]$, namely the top-level marginal of $R_0$, appearing in the latter. In \cite{Mo15} this factor is termed backwards signal of $R_0$, as it distinguishes the directional information flow of DBNs from the undirected one of DBMs.

Under certain assumptions, the backwards signal of $R_0$ can be absorbed in a change of couplings of $R_1$, yielding a system $S'=R_0+R_1'$ whose DBN output approximates the DBM output of $S$.
More precisely, when the top-level marginal of $R_0$ is \emph{approximately spiked}, i.e.\ up to $\mathcal{O}(e^{-\Delta})$ satisfies
\begin{equation}\label{eq:top level marg assumption}
        p_{V_1}[R_0](\mathbf{h}_1) = \begin{cases}
            a  \ &\text{if} \ \mathbf{h}_1 = \mathbf{x} \\ 
             b \ &\text{else},
        \end{cases}
\end{equation}
then there exist systems $R_1'$ and $S'=R_0+R_1'$ such that 
\begin{equation}\label{eq:dbn dbm sim}
          q_{V_0}[S'](\mathbf{h}_0) = p_{V_0}[S](\mathbf{h}_0) + \mathcal{O}(e^{-\Delta}).
\end{equation}
This is proven in Lemma 4 of \hyperref[sec:supplemental]{the supplemental material}. 

For $k>2$ the difference between DBM and DBN output consist of multiple backward signals. 
In Lemma 5 of \hyperref[sec:supplemental]{the supplemental material}, we prove that under certain assumptions, 
they can be absorbed by applying Lemma 4 of \hyperref[sec:supplemental]{the supplemental material} iteratively. Hence, 
some DBMs can be simulated by DBNs. 
We call such DBMs \emph{effectively directional}.

We now prove that effectively directional DBMs are universal approximators (\cref{thm:univ approx}). Since they can be simulated by DBNs, this entails a UAT for DBNs (\cref{thm:DBN UAT}).
To this end, we consider DBMs whose top-level RBM is a linear combination of two flag RBMs for configurations $\mathbf{z}$ and $\mathbf{x}$, 
\begin{equation}
    \centering
    \includegraphics{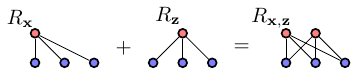}
\end{equation}
When scaled appropriately, $R_{\mathbf{x},\mathbf{z}}$ simulates
\begin{equation}\label{eq:top-level flag Ham}
        f_{\mathbf{x},\mathbf{z}}(\mathbf{v}) =
        \begin{cases}
                    c \ &\text{if} \ \mathbf{v}=\mathbf{x} \\ 
        d \ & \text{if} \  \mathbf{v}=\mathbf{z} \\ 
         \Delta \ & \text{else},
        \end{cases}
\end{equation}
with cut-off $\Delta$.

Each lower level RBM simulates 
\begin{equation}\label{eq:sharing Ham}
    f_{\mathbf{y}_i,\mathbf{x}}(\mathbf{v},\mathbf{h}) = \begin{cases}
         a_i \ &\text{if} \ \mathbf{v}=\mathbf{h} \\ 
        b_i \ & \text{if} \  \mathbf{v}=\mathbf{y}_i \ \text{and} \ \mathbf{h}=\mathbf{x} \\ 
         \Delta \ & \text{else},
    \end{cases}
\end{equation}
with cut-off $\Delta$.
Following \cite{Su08} we call systems that simulate \cref{eq:sharing Ham} \emph{sharing systems} and denote them $S_{\mathbf{y}_i,\mathbf{x}}$. 
Their Boltzmann distribution is such that whenever $\mathbf{h}=\mathbf{x}$ there is a non-zero probability (determined by $a_i$ and $b_i$) to obtain $\mathbf{v}=\mathbf{y}_i$, while all  configurations with $\mathbf{h}\neq \mathbf{x}$ are deterministically copied to $\mathbf{v}=\mathbf{h}$. 
In this sense, $S_{\mathbf{y}_i,\mathbf{x}}$ shares probability mass between $\mathbf{x}$ and $\mathbf{y}_i$.

We construct $S_{\mathbf{y}_i,\mathbf{x}}$ from a copy system, $C$, two flag systems, $R_{\mathbf{y}_i}$ and $R_{\mathbf{x}}$, and an interaction between their flag spins, $E$,
\begin{equation}
    \centering 
    \centering
       \includegraphics{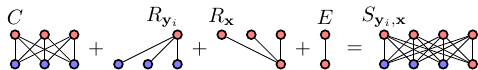}
\end{equation}
In words, $C$ introduces an energy penalty whenever $\mathbf{v} \neq \mathbf{h}$. The flag systems ensure that if $\mathbf{h}=\mathbf{x}$ and $\mathbf{v}=\mathbf{y}_i$, in their low energy sector, both flag spins are in state $1$.
Choosing $E$ appropriately then cancels the energy penalty of $C$ and moves the configuration with $\mathbf{h}=\mathbf{x}$ and $\mathbf{v}=\mathbf{y}_i$ to the overall low energy sector.
Details can be found in Lemma 1 of \hyperref[sec:supplemental]{the supplemental material}.

We now consider DBMs composed of sharing systems and a top level flag RBM.
We call such DBMs \emph{multi-sharing systems}.
\begin{equation}\label{eq:multi-sharing}
    \centering
      \includegraphics{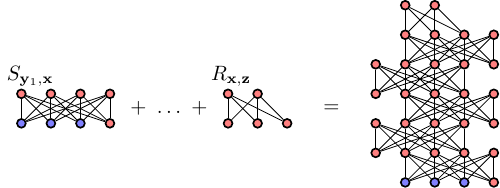}
\end{equation}

In Lemma 5 and Lemma 6 of \hyperref[sec:supplemental]{the supplemental material}, we prove that multi-sharing systems are effectively directional and universal.
This directly implies a UAT for DBNs, and shows that the characterization of universal approximators extends beyond energy based models.
\begin{theorem}\label{thm:DBN UAT}
    DBNs of constant width are universal approximators.
\end{theorem}
Specifically, DBNs with layers of width $n+1$ can simulate arbitrary probability distributions over configurations of $n$ spins.
\begin{proof}
    It follows from Lemma 5 and Lemma 6 of \hyperref[sec:supplemental]{the supplemental material}, and \cref{thm:univ approx}.
\end{proof}

\emph{Outlook}.---
The characterization of universal approximators derived in this work not only unifies existing UATs, but also extends beyond known examples. It may thus guide the construction of universal machine learning models, which could be attractive in two situations. 

First, for ``non-standard" training algorithms. While ``standard" training by contrastive divergence is particularly efficient for bipartite or multipartite architectures \cite{Hi02}, the performance of other training algorithms depends on other properties \cite{So21, Ca25}.
Can we find universal models with the desirable properties?

Second, for non-generic target distributions. 
The characterization of universality implies that certain targets can be efficiently simulated, that is with a polynomial instead of exponential number of parameters. 
Given a set of target distributions, can we devise a universal model that can efficiently represent all distributions from the set? 
Conversely, given a model, what can it efficiently represent?
Even for RBMs and DBMs this is the subject of ongoing research \cite{Ga17, Ca18, Gu20, Ma13}.

Further,
while our results for DBNs extend the characterization of universal approximators beyond energy based models, 
other examples of universal approximators are often based on continuous variables instead of discrete spins.
These include feed forward neural networks (FFNs) \cite{Ho89, Cy89}, graph neural networks \cite{Sa20} and transformers \cite{Yu20}. 
Can our characterization be extended to continuous variables  (cf.\ \cite{De16b}) and thus to some of these examples (cf.\ \cite{Ro10}).

Finally, the mathematical structure underpinning universal spin models is shared with other instances of universality,
including universal Turing machines and universal gate sets from the theory of (quantum) computation \cite{Go24}. 
Since this work puts neural
networks and their feature of universality on an equal footing with these other instances, it raises the following questions.
Do other instances allow for a characterization of universality similar to spin models and neural networks?
How do different instances of universality compare? 
For example, comparing neural networks to Turing machines (or weaker automata) might reveal differences of their computational power \cite{Si95,Pe09} (cf.\ \cite{Re21c}), possibly even shedding light on the distinction between artificial and natural computation \cite{Fr21}.

\emph{Acknowledgements}.---
This research was funded in parts by the Austrian Science Fund 10.55776/Y1261. For open access purposes, the author has applied a CC BY public copyright license to any author accepted manuscript version arising from this submission.


%


\clearpage
\appendix

\crefalias{section}{supplement}
\section*{Supplemental material}\label{sec:supplemental}

\setcounter{theorem}{0}
\setcounter{equation}{0}
\setcounter{figure}{0}
\setcounter{table}{0}
\setcounter{page}{1}

Here we provide details to the results of the main text. 

\begin{lemma}\label{lem:sharing}
Consider the spin system
\begin{equation}
    S_{\mathbf{y}_i,\mathbf{x}} = 2C + \bigl(2\mathrm{d}(\mathbf{y}_i,\mathbf{x})-\lambda_i\bigr) (R_{\mathbf{y}_i}+R_{\mathbf{x}}+E),
\end{equation}
where $C$ is a copy system,  $R_{\mathbf{x}}$ and $R_{\mathbf{y}_i}$ are flag-systems with fields on the flag spins $f_{\mathbf{x}}$ and $f_{\mathbf{y}_i}$ equal to $2\mathrm{d}(\mathbf{x},\mathbf{0})+1$ and $2\mathrm{d}(\mathbf{y}_i,\mathbf{0})+1$, and $E$ is an interaction between $f_{\mathbf{x}}$ and $f_{\mathbf{y}_i}$ with coupling $-3$.
Then, for arbitrary $a_i,b_i$ and $\Delta$, there exist $\delta_i\geq 0$ and $\lambda_i \leq 1$ such that 
$\delta_i S_{\mathbf{y}_i,\mathbf{x}}$ simulates 
\begin{equation}\label{eq:sharing Ham Supp}
    f_{\mathbf{y}_i,\mathbf{x}}(\mathbf{v},\mathbf{h}) = \begin{cases}
         a_i \ &\text{if} \ \mathbf{v}=\mathbf{h} \\ 
        b_i \ & \text{if} \  \mathbf{v}=\mathbf{y}_i \ \text{and} \ \mathbf{h}=\mathbf{x} \\ 
         \Delta \ & \text{else},
    \end{cases}
\end{equation}
i.e.\ yields a sharing system.
\end{lemma}
\begin{proof}
By definition
\begin{widetext}
    \begin{equation}\label{eq:sharing spin system}
    H_{\delta_iS_{\mathbf{y}_i,\mathbf{x}}}(\mathbf{v},\mathbf{h},f_{\mathbf{x}},f_{\mathbf{y}_i}) = 
    \begin{cases}
        \delta_i\lambda_i \ &\text{if} \ \mathbf{v} = \mathbf{y}_i \ \text{and} \ \mathbf{h}=\mathbf{x} \ \text{and} \  f_{\mathbf{x}}=f_{\mathbf{y}_i}=1 \\
        0 \ &\text{if} \ \mathbf{v}=\mathbf{h} \ \text{and} \ f_{\mathbf{x}}=f_{\mathbf{y}_i}=0 \\
        \geq \delta_i \ &\text{else}.
    \end{cases}
\end{equation}
\end{widetext}

We now choose 
$\delta_i = \Delta -b_i$ and $\lambda_i = \tfrac{a_i-b_i}{\Delta - b_i}$.
Without loss of generality we assume that $\Delta$ is large enough such that $\delta_i\geq 0$ and $\lambda_i \leq 1$.
Inserting this into \cref{eq:sharing spin system} proves that $\delta_i S_{\mathbf{y}_i,\mathbf{x}}$ simulates \cref{eq:sharing Ham Supp} with shift $b_i$ and cut-off $\Delta$.
\end{proof}

\begin{lemma}\label{lem:sim approx density}
    Let $S\to T$ with cut-off $\Delta \geq  \max(H_T)$, 
     then for any target configuration $\mathbf{t}$
    \begin{equation}
        p_{V_T}[S](\mathbf{t}) = p[T](\mathbf{t}) + \mc{O}(e^{-\Delta}).\label{eq:probdistr-error}
    \end{equation}
\end{lemma}
\begin{proof}
    Since energy shifts do not change Boltzmann distributions, without loss of generality we assume that the shift between $S$ and $T$ is zero. 
Since $S$ simulates $T$ with cut-off $\Delta\geq \max(H_T)$, for each target configuration $\mathbf{t}$ there exists a unique source configuration $(\mathbf{t},\mathbf{h}_{\mathbf{t}})$ with matching energy. Thus,
\begin{equation}\label{eq:split part fun}
        Z_S = \sum_{\mathbf{t}, \mathbf{h}}e^{-H_S(\mathbf{t},\mathbf{h})} = \sum_{\mathbf{t}}e^{-H_T(\mathbf{t})} + \sum_{(\mathbf{t},\mathbf{h}), \mathbf{h}\neq \mathbf{h}_{\mathbf{t}}}e^{-H_S(\mathbf{t},\mathbf{h})}.
\end{equation}
The first term equals $Z_T$ while the second is upper bounded by $e^{-\Delta}$, 
\begin{equation}
    Z_S = Z_T + \mathcal{O}(e^{-\Delta}).
\end{equation}
Next, we split the sum involved in the definition of the marginal $p_{V_T}[S](\mathbf{t})$ into $(\mathbf{t},\mathbf{h}_{\mathbf{t}})$ and configurations with $\mathbf{h}\neq \mathbf{h}_{\mathbf{t}}$, yielding
\begin{equation}
p_{V_T}[S](\mathbf{t}) = \frac{1}{Z_S} \Bigl( e^{-H_T(\mathbf{t})}+ \mathcal{O}(e^{-\Delta})  \Bigr).
\end{equation}
This results in \eqref{eq:probdistr-error}.
\end{proof}

\begin{lemma}\label{lem:sim prob canonical}
    Let $p$ be any probability distribution on $\{0,1\}^V$ and $T$ be a spin system with Hamiltonian
    \begin{equation}\label{eq:ham prob 2}
        H_T(\mathbf{t}) = \begin{cases}
            -\log(p(\mathbf{t})) 
            & \text{if } p(\mathbf{t})>0 \\
            \Delta             & \text{if }  p(\mathbf{t})=0 .
        \end{cases}
\end{equation}
Then $T$ simulates $p$ with cut-off $\Delta$, i.e.\ for any configuration $\mathbf{t}$,
\begin{equation}
    p[T](\mathbf{t}) = p(\mathbf{t}) + \mathcal{O}(e^{-\Delta}).
\end{equation}
\end{lemma}
\begin{proof}
First, we express the partition function of $T$ as   
\begin{equation}
\begin{split}
    Z_T & 
    = \sum_{\mathbf{t},  p(\mathbf{t})>0} p(\mathbf{t}) + \sum_{\mathbf{t},  p(\mathbf{t})=0}e^{-\Delta} \\
    &= 1 + \mathcal{O}(e^{-\Delta}),
\end{split}
\end{equation}
where the first equality follows by the definition of $H_T$ and the second from the normalization of $p$.
Inserting it in the definition of the Boltzmann distribution of $T$, $p[T]$, proves the claim.
\end{proof}

\begin{lemma}\label{lem:DBN DBM 2}
    Let $S=R_0+R_1$ such that the top-level marginal of $R_0$ is approximately spiked. 
    Then there exists 
    a spin system $S'=R_0+R_1'$ such that 
    \begin{equation}\label{eq:claim DBN DBM}
         q_{V_0}[S'](\mathbf{h}_0) = p_{V_0}[S](\mathbf{h}_0)  + \mathcal{O}(e^{-\Delta}).
    \end{equation}
\end{lemma}
\begin{proof}
By assumption the top-level marginal of $R_0$ satisfies
\begin{equation}
        p_{V_1}[R_0](\mathbf{h}_1) = \begin{cases}
            a  \ &\text{if} \ \mathbf{h}_1 = \mathbf{x} \\ 
             b \ &\text{else},
        \end{cases}
\end{equation}
for some configuration $\mathbf{x}$ and real numbers $a,b$, up to $\mathcal{O}(e^{-\Delta})$.
We construct $R_1'$ by changing couplings of $R_1$ such that below $\Delta$ 
\begin{equation}\label{eq:change top RBM}
        H_{R_1'}(\mathbf{h}_1,\mathbf{h}_2) = 
        \begin{cases}
            H_{R_1}(\mathbf{h}_1,\mathbf{h}_2)- c \ &\text{if} \ \mathbf{h}_1 = \mathbf{x} \\ 
             H_{R_1}(\mathbf{h}_1,\mathbf{h}_2) \ &\text{else}.
        \end{cases}
\end{equation}
We now prove that an appropriate choice of $c$ (see \cref{eq:choice f}) leads to 
\begin{equation}\label{eq:what to prove}
        p_{V_1}[R_1'](\mathbf{h}_1) = (p_{V_1}[R_0]\ast p_{V_1}[R_1])(\mathbf{h}_1) + \mathcal{O}(e^{-\Delta}).
\end{equation}
Since 
\begin{equation}
    q_{V_0}[S](\mathbf{h}_0) = \sum_{\mathbf{h}_1}p[R_0](\mathbf{h}_0 \, \vert \, \mathbf{h}_1) p_{V_1}[R_1](\mathbf{h}_1).
\end{equation}
and 
\begin{equation}
     p_{V_0}[S](\mathbf{h}_0) = \sum_{\mathbf{h}_1}p[R_0](\mathbf{h}_0 \, \vert \, \mathbf{h}_1) (p_{V_1}[R_0]\ast p_{V_1}[R_1])(\mathbf{h}_1),
\end{equation}
this proves the claim.

First, by definition of $R_1'$, up to $\mathcal{O}(e^{-\Delta})$, 
    \begin{equation}\label{eq:change top level}
        p_{V_1}[R_1'](\mathbf{h}_1) = \begin{cases}
           d  p_{V_1}[R_1](\mathbf{h}_1) \ & \text{if} \ \mathbf{h}_1=\mathbf{x} \\ 
            d' p_{V_1}[R_1](\mathbf{h}_1)   \ & \text {else}
    \end{cases}
\end{equation}
where 
$d=\tfrac{Z}{Z'}e^c$ and $d'=\tfrac{Z}{Z'}$, with $Z$ and $Z'$ the partition functions of $R_1$ and $R_1'$ respectively. 
We obtain
\begin{equation}
    Z' = Z + (e^c-1) Z p_{V_1}[R_1](\mathbf{x}).
\end{equation}

Crucially, as long as
\begin{equation}
    0\leq d p_{V_1}[R_1](\mathbf{x})<1,
\end{equation}
any value of $d$ can be achieved by an appropriate choice of $c$, namely 
\begin{equation}
    c = \log \left( \frac{d-dp_{V_1}[R_1](\mathbf{x})}{1-dp_{V_1}[R_1](\mathbf{x})}\right).
\end{equation}
We choose $c$ such that 
\begin{equation}\label{eq:choice f}
    d = \frac{a}{b(1-p_{V_1}[R_1](\mathbf{x}))+a p_{V_1}[R_1](\mathbf{x})}.
\end{equation}

We first rewrite the normalization constant of $p_{V_1}[R_0]\ast p_{V_1}[R_1]$ as 
\begin{equation}
\begin{split}
        N &= \sum_{\mathbf{h}_1}p_{V_1}[R_0](\mathbf{h}_1) p_{V_1}[R_1](\mathbf{h}_1) \\
        &= b\sum_{\mathbf{h}_1\neq \mathbf{x}} p_{V_1}[R_1](\mathbf{h}_1) + a p_{V_1}[R_1](\mathbf{x}) + \mathcal{O}(e^{-\Delta}) \\
        &= \frac{a}{d} + \mathcal{O}(e^{-\Delta}),
\end{split}
\end{equation}
where the second equality holds since the top-level marginal of $R_0$ is approximately spiked, and the third because  $p_{V_1}[R_1]$ is normalized.
It follows that 
\begin{equation}
    (p_{V_1}[R_0]\ast p_{V_1}[R_1])(\mathbf{x}) = d p_{V_1}[R_1](\mathbf{x}) + \mathcal{O}(e^{-\Delta})
\end{equation}
which proves \cref{eq:what to prove} for the case $\mathbf{h}_1=\mathbf{x}$.

Finally, normalization of $p_{V_1}[R_1']$ implies
\begin{equation}
    1 = d p_{V_1}[R_1](\mathbf{x}) + d'\sum_{\mathbf{h}_1 \neq \mathbf{x}}p_{V_1}[R_1](\mathbf{h}_1).
\end{equation}
Normalization of $p_{V_1}[R_0]\ast p_{V_1}[R_1]$ yields 
\begin{equation}
    1 = d p_{V_1}[R_1](\mathbf{x}) + \frac{d}{a} b \sum_{\mathbf{h}_1 \neq \mathbf{x}}p_{V_1}[R_1](\mathbf{h}_1) + \mathcal{O}(e^{-\Delta}).
\end{equation}
Thus, up to $\mathcal{O}(e^{-\Delta})$, 
$d'=\frac{d}{a}b$. This proves \cref{eq:what to prove} for the case $\mathbf{h}_1\neq\mathbf{x}$.
\end{proof}

\begin{lemma}\label{lem:DBM-DBN sim}
    Multi-sharing systems 
    are effectively directional.
\end{lemma}
\begin{proof} 
We decompose a multi-sharing system $S$ into the sum of its bottom most sharing system $S_{\mathbf{y}_1,\mathbf{x}}$ and the remaining system
\begin{equation}
W =S_{\mathbf{y}_2,\mathbf{x}}+\ldots + R_{\mathbf{x},\mathbf{z}}.
\end{equation}
It is easy to see that sharing systems have approximately spiked top-level marginals. 
We hence apply \cref{lem:DBN DBM 2} to $S= S_{\mathbf{y}_1,\mathbf{x}} +W$, yielding a system $W'$ that absorbs the backwards signal of $S_{\mathbf{y}_1,\mathbf{x}}$. 
By \cref{eq:change top RBM}, $W$ and $W'$ only differ on configurations with $\mathbf{h}_1=\mathbf{x}$. 
This modification can be achieved by only changing couplings of $R_{\mathbf{x},\mathbf{z}}$,
i.e. 
\begin{equation}
W'=S_{\mathbf{y}_2,\mathbf{x}}+\ldots + R_{\mathbf{x},\mathbf{z}}'. 
\end{equation}
In particular also $W'$ is a multi-sharing system and we can 
apply \cref{lem:DBN DBM 2} to the split of $W $ into $S_{\mathbf{y}_2,\mathbf{x}}$ and $S_{\mathbf{y}_3,\mathbf{x}}+\ldots+R_{\mathbf{x},\mathbf{z}}'$, yielding a system 
\begin{equation}
W''=S_{\mathbf{y}_3,\mathbf{x}}+\ldots+R_{\mathbf{x},\mathbf{z}}''
\end{equation}
that absorbs the backwards signal of $S_{\mathbf{y}_2,\mathbf{x}}$.
Proceeding iteratively proves the claim.
\end{proof}

\begin{lemma}\label{lem:DBM univ mod}
    The spin model consisting of multi-sharing systems is universal.
\end{lemma}
More precisely, we prove that multi-sharing systems can simulate arbitrary target systems such that physical spins correspond to visible spins.
\begin{proof}
    We construct a simulation of a generic function $g\colon \{0,1\}^{V_0} \to \mathbb{R}$ such that visible spins of the multi-sharing system correspond to physical spins. 
    Universality follows from choosing $g$ to be the Hamiltonian of a generic target system $T$. As noted before, it may be more efficient to construct the simulation of $T$ by adding simulations of all of its local energy functions  as opposed to directly simulating $H_T$. 

    Recall that $S_{\mathbf{y}_i,\mathbf{x}}$ simulates $f_{\mathbf{y}_i,\mathbf{x}}$ defined in \cref{eq:sharing Ham Supp}, and  $R_{\mathbf{x},\mathbf{z}}$ simulates 
    \begin{equation}\label{eq:top-level flag Ham Supp}
        f_{\mathbf{x},\mathbf{z}}(\mathbf{v}) =
        \begin{cases}
                    c \ &\text{if} \ \mathbf{v}=\mathbf{x} \\ 
        d \ & \text{if} \  \mathbf{v}=\mathbf{z} \\ 
         \Delta \ & \text{else}.
        \end{cases}
\end{equation}
    Adding them  implies that 
    $S = S_{\mathbf{y}_1,\mathbf{x}}+\ldots+S_{\mathbf{y}_k,\mathbf{x}}+R_{\mathbf{x},\mathbf{z}}$ simulates 
    \begin{widetext}
\begin{equation}\label{eq:ham sharing S}
    f(\mathbf{h}_0, \ldots, \mathbf{h}_k) = \begin{cases}
        c + \sum_i a_i\ & \text{if} \ \mathbf{h}_0=\ldots = \mathbf{h}_k=\mathbf{x}\\
        d + \sum_i a_i \ &\text{if} \  \mathbf{h}_0= \ldots =\mathbf{h}_k=\mathbf{z} \\
        c +  b_j+\sum_{i\neq j} a_i \ & \text{if} \ \mathbf{h}_0=\ldots =\mathbf{h}_{j-1}=\mathbf{y}_j \ \text{and} \ \mathbf{h}_j= \ldots = \mathbf{h}_k=\mathbf{x}\\
        \Delta \ &\text{else}. 
    \end{cases}
\end{equation}
\end{widetext}

    We now enumerate low-energy configurations of the target $g$ (i.e.\ configurations with $g(\mathbf{v})\leq \Delta$). 
    We identify the first one, $\mathbf{v}_1$, with $\mathbf{x}$, 
    the second one, $\mathbf{v}_2$, with $\mathbf{z}$, and the remaining ones such that $\mathbf{v}_{j+2}$ corresponds to $\mathbf{y}_j$. 
    Since parameters $c,d,a_j,b_j$ can be chosen arbitrarily via the scalings and couplings of $S_{\mathbf{y}_i,\mathbf{x}}$ and $R_{\mathbf{x},\mathbf{z}}$,  
    we choose $a_i = 0$, $c = g(\mathbf{v}_1)$, $d =  g(\mathbf{v}_2)$ and $b_j=g(\mathbf{v}_{j+2})-g(\mathbf{v}_1)$.
    Inserting this into \cref{eq:ham sharing S}
    proves that $S$ simulates $g$ with cut-off $\Delta$.  
\end{proof}

\end{document}